\documentclass[12pt]{article}
\usepackage[T1]{fontenc} 
\usepackage{fullpage}
\usepackage{microtype}
\usepackage{amsmath}
\usepackage{amsthm}
    \newtheorem*{theorem*}{Theorem}
    \newtheorem{theorem}{Theorem}[section]
    \newtheorem{lemma}{Lemma}[section]
    
    \newtheorem{conjecture}{Conjecture}
\usepackage{graphicx} 
    \graphicspath{{./graphics/}}
    \newcommand{\graphicsScale}{1.}
\usepackage{tabularx}
    \newcolumntype{C}{@{}>{\centering\arraybackslash}X@{}}
\usepackage{enumitem}
    \newlist{thmEnumerate}{enumerate}{1}
    \setlist[thmEnumerate, 1]{label=(\alph{thmEnumeratei})}
    
    \newlist{Cases}{enumerate}{5}
    \setlist[Cases]{wide}
    \setlist[Cases, 1]{label=\textbf{Case~\arabic{Casesi}:},ref=\arabic*}
    \setlist[Cases, 2]{label=\textbf{Case~\arabic{Casesi}.\arabic{Casesii}:},ref=\arabic*}
    \setlist[Cases, 3]{label=\textbf{Case~\arabic{Casesi}.\arabic{Casesii}.\arabic{Casesiii}:},ref=\arabic*}
    \setlist[Cases, 4]{label=\textbf{Case~\arabic{Casesi}.\arabic{Casesii}.\arabic{Casesiii}.\arabic{Casesiv}:},ref=\arabic*}
    \setlist[Cases, 5]{label=\textbf{Case~\arabic{Casesi}.\arabic{Casesii}.\arabic{Casesiii}.\arabic{Casesiv}.\arabic{Casesv}:},ref=\arabic*}

    \newlist{inlineEnum}{enumerate*}{1}
    \setlist[inlineEnum, 1]{label=(\roman{inlineEnumi})}
\usepackage{hyperref}
\usepackage{color}
\newcommand{\FlipGraph}[1]{\mathcal{G}(#1)} 
\newcommand{\Vertices}[1]{V(#1)} 
\newcommand{\Edges}[1]{E(#1)} 
\newcommand{\OO}[1]{O(#1)} 
\newcommand{\mya}{a} 

\title{Further Connectivity Results on\\ Plane Spanning Path Reconfiguration}
\author{
  Valentino Boucard\thanks{Aix-Marseille Université and LIS, France. The work of Guilherme D. da Fonseca is supported by the French ANR PRC grant ADDS (ANR-19-CE48-0005).} 
  \and
  Guilherme D. da Fonseca\textsuperscript{*}  
  \and 
  Bastien Rivier\thanks{Brock University, Canada.}
  }
\date{}

\begin{document}

\maketitle

\begin{abstract}
Given a finite set $ S $ of points, we consider the following reconfiguration graph. The vertices are the plane spanning paths of $ S $ and there is an edge between two vertices if the two corresponding paths differ by two edges (one removed, one added). Since 2007, this graph is conjectured to be connected but no proof has been found. In this paper, we prove several results to support the conjecture. Mainly, we show that if all but one point of $ S $ are in convex position, then the graph is connected with diameter at most $ 2 | S | $
and that for $ | S | \geq 3 $ every connected component has at least $ 3 $ vertices.
\end{abstract}

\section{Introduction}

Reconfiguration problems consider the sequence of operations used to convert a certain \emph{configuration} into another and can be formalized using the following (possibly directed) graph called the \emph{flip graph}. The vertices of the graph represent the configurations in question. The edges are defined by a \emph{flip} operation that locally changes a configuration, producing a new one. The most fundamental question is whether the corresponding graph is connected, that is, whether, for any two configurations, we are able to reconfigure one into the other using a sequence of flips.

In this paper, we are interested in the connectivity of the undirected flip graph $ \FlipGraph { S } $ defined hereafter.
Let $ S $ be a set of $ n \geq 3 $ points in general position (namely, no three points are colinear).
We define the \emph{flip graph} of $ S $, denoted $ \FlipGraph { S } $, as the following undirected simple graph.
The vertices $ \Vertices { \FlipGraph { S } } $, called \emph{plane paths}, are the spanning paths drawn in the plane with non-crossing straight line segments.
The edges $ \Edges {  \FlipGraph { S } } $, called \emph{flips}, consist of pairs of plane paths whose symmetric difference is made of exactly two segments (see Figure~\ref{fig:Demo}). We refer to a path in $ \FlipGraph { S } $ as a \emph{flip path}.

\begin{figure}[htb]
    \centering
    \begin{tabularx}{\textwidth}{C}
        \includegraphics[scale=\graphicsScale]{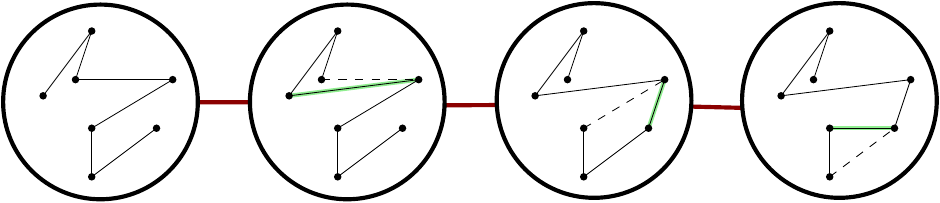}
    \end{tabularx}
\caption{Example of a flip path in $\FlipGraph{S}$. From left to right, the removed segment is dashed and the inserted segment is highlighted.}
 \label{fig:Demo}
\end{figure}

The main open problem about $ \FlipGraph { S } $ is to resolve the following conjecture, first proposed in~\cite{On}, and further studied in~\cite{Flipping}.

\begin{conjecture}[\cite{On}]
\label{cjt:main}
    For every point set $ S $ in general position, the flip graph $ \FlipGraph { S } $ is connected.
\end{conjecture}

The following theorems prove Conjecture~\ref{cjt:main} for special classes of point sets.

\begin{theorem}[\cite{On}]
    \label{thm:convex}
    For every point set $ S $ in convex position, the flip graph $ \FlipGraph { S } $ is connected with diameter at most $ 2 n - 6 $.
\end{theorem}

\begin{theorem}[\cite{Flipping}]
    \label{thm:1In}
    For every point set $ S = C \cup \{ \xi  \} $ where $ C $ is in convex position and $ \xi  $ lies inside the convex hull of $ C $, the flip graph $ \FlipGraph { S } $ is connected with diameter at most $ 2 n - 4 $.
\end{theorem}

\begin{theorem}[\cite{Flipping}]
    \label{thm:doubleCircle}
    For every point set $ S $ in \emph{generalized double circle position} (a type of deformation of convex position), the flip graph $ \FlipGraph { S } $ is connected with diameter at most $ \OO { n ^ 2 } $.
\end{theorem}

\paragraph{Contribution.}
In this paper, we prove the following theorem\footnote{Simultaneously to our work, Kleist et al.~\cite{KKR24} independently proved a result that implies that the flip graph in Theorem~\ref{thm:1InOrOut}(a) is connected. We are not aware of the diameter they proved or if Theorem~\ref{thm:1InOrOut}(b) and (c) also follow from their work.}. 

\begin{theorem*}[\ref{thm:1InOrOut}]
Let $ C $ be a set of $ n - 1 $ points in convex position and $ \xi $ be a point outside the convex hull of $ C $. The following holds.
\begin{thmEnumerate}
    \item The flip graph of $ S = C \cup \{ \xi \} $ is connected and its diameter is at most $ 2 n $.
    \item The subgraph of the flip graph of $ S = C \cup \{ \xi \} $ induced by paths where $ \xi $ has degree $ 1 $ is connected and its diameter is at most $ 4 n - 15 $.
    \item The subgraph of the flip graph of $ S = C \cup \{ \xi \}$ induced by plane paths where $ \xi $ has degree $ 2 $ is connected and its diameter is at most $ 2 n $.
\end{thmEnumerate}
\end{theorem*}

Furthermore, we show that, for every set $ S $ of at least $ 3 $ points in general position, any connected component of the flip graph $ \FlipGraph { S } $ has at least $ 3 $ vertices. 

\paragraph{Related Work.}
Flip graphs of non-crossing segments with a flip operation that removes a segment and inserts another one have been studied for several structures other than paths.
In the case of triangulations~\cite{lubiw2015flip,pilz2014flip} the graph is connected with diameter $ \OO { n ^ 2 } $ for $ n $ points in general position~\cite{lawson1972transforming} and at most $ 2 n - 10 $ for $ n \geq 13 $ points in convex position~\cite{pournin2014diameter}.
In the case of trees~\cite{aichholzer2022reconfiguration} the flip graph is known to be connected with diameter at most $ 2 n - 4 $ for $ n $ points in general position~\cite{nichols2020transition} and at most $ \frac{5}{3} n - 3 $ for $ n $ points in convex position~\cite{newtrees}.

The connectivity of the flip graph is open for non-crossing matchings~\cite{hernando2002graphs, houle2005graphs}, but the flip operation consists of exchanging the edges in a cycle. If we allow segments to cross, a natural flip operation consists of exchanging pairs of crossing segments by non-crossing segments with the same endpoints. Using this operation, we always reach a non-crossing set of segments before a cubic number of flips~\cite{VLSC81}, but better bounds exist for special point sets such as convex~\cite{OdW07} and near-convex~\cite{longestuntangleconf}.

For a comprehensive overview of reconfiguration problems, readers can refer to multiple surveys~\cite{nishimura2018introduction,van2013complexity}.

\section{One Point Outside a Convex}

In this section, we show that Conjecture~\ref{cjt:main} holds for a set of points if all points but one are in convex position.
Conjecture~\ref{cjt:main} has already been proved in~\cite{Flipping} in the case where the point which is not in convex position lies in the interior of the convex hull of the point set. We provide a proof for the case where the point which is not in convex position lies outside the convex hull of the other points.

\begin{theorem}
\label{thm:1InOrOut}
Let $ C $ be a set of $ n - 1 $ points in convex position and $ \xi $ be a point outside the convex hull of $ C $.
The following holds.
\begin{thmEnumerate}
    \item \label{item:1INthm:1InOrOut} The flip graph of $ S = C \cup \{ \xi \} $ is connected and its diameter is at most $ 2 n $.
    \item \label{item:2INthm:1InOrOut} The subgraph of the flip graph of $ S = C \cup \{ \xi \} $ induced by paths where $ \xi $ has degree $ 1 $ is connected and its diameter is at most $ 4 n - 15 $.
    \item \label{item:3INthm:1InOrOut} The subgraph of the flip graph of $ S = C \cup \{ \xi \}$ induced by plane paths where $ \xi $ has degree $ 2 $ is connected and its diameter is at most $ 2 n $.
\end{thmEnumerate}
\end{theorem}

The proof of Theorem~\ref{thm:1InOrOut}\ref{item:1INthm:1InOrOut} is based on two types of plane paths: canonical and non-canonical (defined hereafter).
In turn, a canonical path may be either strongly canonical or not (defined hereafter).
In Lemma~\ref{lem:canonical}, we show that any canonical path is connected to a strongly canonical path.
In Lemma~\ref{lem:main}, we show that any non-canonical path is connected to some canonical path.
The two lemmas together immediately imply Theorem~\ref{thm:1InOrOut}\ref{item:1INthm:1InOrOut}.
In the following, we define the notions used in the proofs of Lemma~\ref{lem:canonical}, Lemma~\ref{lem:main}, and Theorem~\ref{thm:1InOrOut}.

Given a set of points $ C $ in convex position and a point $ \xi $ outside the convex hull of $ C $, a \emph{canonical segment} of $ C \cup \{ \xi \} $ is a segment that is either incident to $ \xi $ or is an edge of the convex hull of $ C $.
We classify the canonical segments into three types: the edges of the convex hull of $ C $, called \emph{convex} segments, the segments incident to $ \xi $ which do not intersect the interior of the convex hull of $ C $, called \emph{outer} segments, and the remaining segments, called \emph{inner} segments.
A \emph{canonical path} of $ C \cup \{ \xi \} $ is then defined as a plane path of $ C \cup \{ \xi \} $ that contains only canonical segments.
We also define a special type of canonical paths. A canonical path $ P _ 0 $ is \emph{strongly canonical} if (i) $ \xi $ is not an extremity of $ P $, (ii) the two points $ \alpha _ 0 , \beta _ 0 $ adjacent to $ \xi $ in $ P $ are consecutive on the convex hull of $ C $, and (iii) the segments $ \xi \alpha _ 0 , \xi \beta _ 0 $ are outer segments.
We now state and prove Lemma~\ref{lem:canonical}.

\begin{lemma}
    \label{lem:canonical}
    There exists a path of length at most $ 6 $ which connects any canonical path $ P $ to any strongly canonical path $ P _ 0 $ in the subgraph of $ \FlipGraph { S } $ induced by canonical paths.
\end{lemma}

\begin{figure}[htb]
    \centering
    \begin{tabularx}{\textwidth}{CCCC}
        \includegraphics[scale=\graphicsScale,page=1]{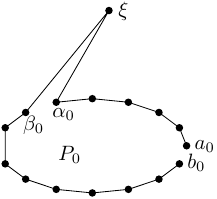}&%
        \includegraphics[scale=\graphicsScale,page=2]{graphics/lemmaCanonical.pdf}&%
        \includegraphics[scale=\graphicsScale,page=3]{graphics/lemmaCanonical.pdf}&%
        \includegraphics[scale=\graphicsScale,page=4]{graphics/lemmaCanonical.pdf}\\
        (a) & (b) & (c) & (d)%
    \end{tabularx}
 \caption{Illustration of the proof of Lemma~\ref{lem:canonical}. (a) The strongly canonical path $ P _ 0 $. (b) The canonical path $ P $ in Case~\ref{case:1}. (c) The canonical path $ P $ in Case~\ref{case:2}. (d) The canonical path $ P _ 1 $ in Case~\ref{case:2}.}
 \label{fig:lemmaCanonicalCase1-2}
\end{figure}

\begin{proof}
    First, note that replacing a segment by an outer segment in a canonical path yields a canonical path since no segment crosses an outer segment.
    Similarly, replacing a segment by a convex segment in a canonical path without inner segments yields a canonical path.
    In fact, throughout the following proof of Lemma~\ref{lem:canonical}, we only perform flips of these two types.

    Let $ P $ (respectively $ P _ 0 $) be an arbitrary canonical path (respectively an arbitrary strongly canonical path).
    Let $ \mya , b $ (respectively $ \mya _ 0 , b _ 0 $) be the extremities of $ P $ (respectively of $ P _ 0 $).
    Let $ \alpha , \beta $ (respectively $ \alpha _ 0 , \beta _ 0 $) be the points adjacent to $ \xi $ in the cycle formed by adding an edge to $ P $ (respectively to $ P _ 0 $).
    Note that, since $ \alpha _ 0 , \beta _ 0 $ are consecutive vertices of the convex hull of $ C $, it is also the case for the extremities $ \mya _ 0 , b _ 0 $ of $ P $.
    Figure~\ref{fig:lemmaCanonicalCase1-2}(a) shows the notations for $ P _ 0 $ while Figure~\ref{fig:lemmaCanonicalCase1-2}(c) shows the notations for $ P $.

    Assume that $ P \neq P _ 0 $.
    In each of the following cases, we construct a path of length at most $ 6 $ in the subgraph of the flip graph induced by canonical paths that connects $ P $ and $ P _ 0 $.
    \begin{Cases}
        \item \label{case:1} \textbf{Moving the hole.}
        If $ \{ \alpha , \beta \} = \{ \alpha _ 0 , \beta _ 0 \} $, then replacing the segment $ \mya _ 0 b _ 0 $ by $ \mya b $ in $ P $ yields $ P _ 0 $ (Figure~\ref{fig:lemmaCanonicalCase1-2}(b)).
        We have a path of length $ 1 $ connecting $ P $ and $ P _ 0 $.
        
        \item \label{case:2} \textbf{Moving the spike.}
        Otherwise, if $ \xi \alpha $ and $ \xi \beta $ are two outer segments (not necessarily in $ P $ or $ P _ 0 $), and if $ \alpha , \beta $ are consecutive on the convex hull of $ C $, then by Case~\ref{case:1} there exists a path of length $ 1 $ connecting $ P $ and a canonical path $ P _ 1 $ preserving the same points adjacent to $ \xi $ but whose extremities are $ \{ \alpha _ 0 , \beta _ 0 \} $ (Figure~\ref{fig:lemmaCanonicalCase1-2}(c) and~(d)).
        Without loss of generality, assume that the sub-path connecting $ \alpha $ and $ \alpha _ 0 $ (respectively $ \beta $ and $ \beta _ 0 $) in $ P _ 1 $ does not contain $ \xi $.
        We perform the following flips on $ P _ 1 $ (either $ 1 $ or $ 2 $ flips are performed).
        If $ \alpha \neq \alpha _ 0 $, then replace the segment $ \xi \alpha $ by $ \xi \alpha _ 0 $.
        If $ \beta \neq \beta _ 0 $, then replace the segment $ \xi \beta $ by $ \xi \beta _ 0 $.
        These flips yield a canonical path $ P _ 2 $ whose points adjacent to $ \xi $ are the same as in $ P _ 0 $ (Figure~\ref{fig:lemmaCanonicalCase2-4}(a)).
        By Case~\ref{case:1}, there exists a path of length $ 1 $ connecting $ P _ 2 $ and $ P _ 0 $.
        Overall, we have a path of length at most $ 4 $ connecting $ P $ and $ P _ 0 $.

        \begin{figure}[htb]
            \centering
            \begin{tabularx}{\textwidth}{CCCC}
                \includegraphics[scale=\graphicsScale,page=5]{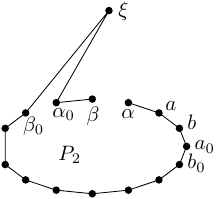}&%
                \includegraphics[scale=\graphicsScale,page=6]{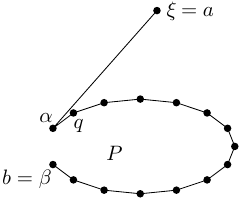}&%
                \includegraphics[scale=\graphicsScale,page=7]{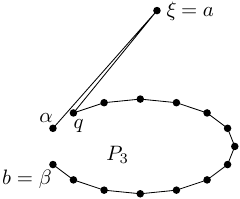}&%
                \includegraphics[scale=\graphicsScale,page=8]{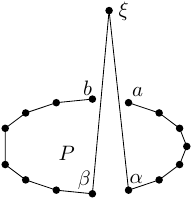}\\
                (a) & (b) & (c) & (d)%
            \end{tabularx}
         \caption{Illustration of the proof of Lemma~\ref{lem:canonical}. (a) The canonical path $ P _ 2 $ in Case~\ref{case:2}. (b) The canonical path $ P $ in Case~\ref{case:3}. (c) The canonical path $ P _ 3 $ in Case~\ref{case:3}. (d) The canonical path $ P $ in Case~\ref{case:4}.}
         \label{fig:lemmaCanonicalCase2-4}
        \end{figure}
        
        \item \label{case:3} \textbf{Moving a phantom inner segment out.} 
        Otherwise, if one of $ \xi \alpha $ or $ \xi \beta $, say $ \xi \beta $, is an inner segment and is not in $ P $, then $ \xi \alpha $ is an outer segment and is in $ P $ (Figure~\ref{fig:lemmaCanonicalCase2-4}(b)).
        In this case, since $ P $ is plane, $ \alpha , \beta $ are consecutive on the convex hull of $ C $.
        Let $ q \neq \beta $ be such that $ q , \alpha $ are consecutive on the convex hull of $ C $.
        The segment $ q \xi $ (not in $ P $) is an outer segment.
        Therefore, replacing the segment $ q \alpha $ by $ q \xi $ in $ P $ yields a canonical path $ P _ 3 $ satisfying the conditions of Case~\ref{case:2} (Figure~\ref{fig:lemmaCanonicalCase2-4}(c)).
        By Case~\ref{case:2}, there exists a path of length at most $ 4 $ connecting $ P _ 3 $ and $ P _ 0 $.
        Overall, we have a path of length at most $ 5 $ connecting $ P $ and $ P _ 0 $.

        \item \label{case:4} \textbf{Moving inner segments out.}
        Otherwise, if at least one of $ \xi \alpha $ or $ \xi \beta $, say $ \xi \alpha $, is an inner segment and is in $ P $, then $ \xi \beta $ is also in $ P $ (and may be either an inner or an outer segment; see Figure~\ref{fig:lemmaCanonicalCase2-4}(d)).
        Without loss of generality, assume that the sub-path of $ P $ connecting $ \mya $ and $ \alpha $ (respectively $ b $ and $ \beta $) does not not contain $ \xi $.
        In this case, since $ P $ is plane, both $ \mya $ and $ b $ are vertices of the convex hull of $ C $, and $ \xi \mya $ is an outer segment (not in $ P $).
        In this setting, replacing the segment $ \xi \alpha $ by $ \xi \mya $ in $ P $ yields a canonical path $ P _ 4 $ with at most $ 1 $ inner segment (Figure~\ref{fig:lemmaCanonicalCase4-5}(a)).
        If $ P _ 4 $ indeed has an inner segment, then going through Case~\ref{case:4} again yields a canonical path $ P _ 5 $ (Figure~\ref{fig:lemmaCanonicalCase4-5}(b)).
        Otherwise, let $ P _ 5 = P _ 4 $.
        In both cases, $ P _ 5 $ satisfies the conditions of Case~\ref{case:2}.
        By Case~\ref{case:2}, there exists a path of length at most $ 4 $ connecting $ P _ 5 $ and $ P _ 0 $.
        Overall, we have a path of length at most $ 6 $ connecting $ P $ and $ P _ 0 $.

        \begin{figure}[htb]
            \centering
            \begin{tabularx}{\textwidth}{CCCC}
                \includegraphics[scale=\graphicsScale,page=9]{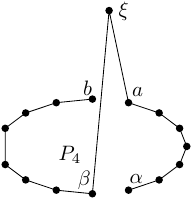}&%
                \includegraphics[scale=\graphicsScale,page=10]{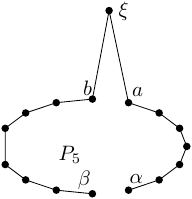}&%
                \includegraphics[scale=\graphicsScale,page=11]{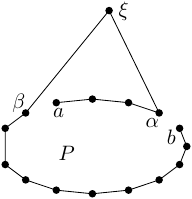}&%
                \includegraphics[scale=\graphicsScale,page=12]{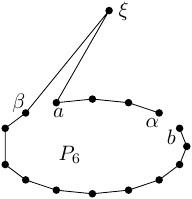}\\
                (a) & (b) & (c) & (d)%
            \end{tabularx}
         \caption{Illustration of the proof of Lemma~\ref{lem:canonical}. (a) The canonical path $ P _ 4 $ in Case~\ref{case:4}. (b) The canonical path $ P _ 5 $ in Case~\ref{case:4}. (c) The canonical path $ P $ in Case~\ref{case:5}. (d) The canonical path $ P _ 6 $ in Case~\ref{case:5}.}
         \label{fig:lemmaCanonicalCase4-5}
        \end{figure}

        \item \label{case:5} \textbf{Shrinking the spike.}
        In this remaining case, $ \alpha , \beta $ are not consecutive on the convex hull of $ C $ (Figure~\ref{fig:lemmaCanonicalCase4-5}(c)).
        Since $ P $ is plane, both $ \xi \alpha $ and $ \xi \beta $ are outer segments (in $ P $).
        Moreover, both $ \mya $ and $ b $ are vertices of the convex hull of $ C $.
        Thus, one of $ \xi \mya $ or $ \xi b $ is an outer segment (not in $ P $).
        This allows the following assumptions.
        We assume without loss of generality that the sub-path connecting $ \alpha $ and $ \mya $ in $ P $ does not contain $ \xi $, and that $ \xi \mya $ is an outer segment.
        In this setting, replacing the segment $ \xi \alpha $ by $ \xi \mya $ in $ P $ yields a canonical path $ P _ 6 $ satisfying the conditions of Case~\ref{case:2} (Figure~\ref{fig:lemmaCanonicalCase4-5}(d)).
        By Case~\ref{case:2}, there exists a path of length at most $ 4 $ connecting $ P _ 6 $ and $ P _ 0 $.
        Overall, we have a path of length at most $ 5 $ connecting $ P $ and $ P _ 0 $. \qedhere
    \end{Cases}
\end{proof}

\begin{lemma}
    \label{lem:main}
    Let $ C $ be a set of $ n - 1 $ points in convex position and $ \xi $ be a point outside the convex hull of $ C $.  For every vertex $ P $ of the flip graph of $ S = C \cup \{ \xi \} $, there exists a flip path of length at most $ n - 6 $ connecting $ P $ and a canonical path.
\end{lemma}
\begin{proof}
    Let $ P = p _ 1 , \dots , p _ n $ be a plane path in $ \Vertices {  \FlipGraph { S } } $.
    Let $ i $ be the index such that $ \xi = p _ i $.
    Next, we describe a flip path of length at most $ i - 4 $ connecting the plane path $ p _ 1 , \dots , p _ { i - 1 } $ to a plane path whose segments are all canonical segments of $ C \cup \{ \xi \} $.
    Since the interior of the convex hull of the plane path $ p _ 1 , \dots , p _ { i - 1 } $ does not intersect the other segments of $ P $, the flip path we describe is not only a flip path in $ \FlipGraph { \{ p _ 1 , \dots , p _ { i - 1 } \} } $ but also a flip path in $ \FlipGraph { S } $.

    If $ i \in \{ 1 , 2 \} $, then the only potential segment of the plane path $ p _ 1 , \dots , p _ { i - 1 } $ is already canonical.
    Otherwise $ i $ is at least $ 3 $, and Lemma~\ref{lem:convex} applied to the plane path $ p _ 1 , \dots , p _ { i - 1 } $ yields a plane path with $ p _ { i - 1 } $ as one of its extremities and with $ p _ { i - 2 } p _ { i - 1 } $  as one of its segments using at most $ i - 4 $ flips.
    If the segment $ p _ { i - 2 } p _ { i - 1 } $ is not canonical, then the extra flip replacing $ p _ { i - 2 } p _ { i - 1 } $ by $ p _ 1 p _ { i - 1 } $ finally leads to a path with only canonical edges.
    Thus, we have used at most $ i - 3 $ flips.

    Similarly, there exists a flip path in $ \FlipGraph { S } $ of length at most $ n - i - 3 $ connecting the plane path $ p _ { i + 1 } , \dots , p _ n $ to a plane path whose segments are all canonical segments of $ C \cup \{ \xi \} $.
    The remaining segments in $ P $ are incident to $ \xi = p _ i $ and thereby already canonical.
    In total, we have used at most $ ( i - 3 ) + ( n - i - 3 ) = n - 6 $, concluding the proof of Lemma~\ref{lem:main}.
\end{proof}

\begin{proof}[Proof of Theorem~\ref{thm:1InOrOut}]
    In the following, we prove the three assertions of Theorem~\ref{thm:1InOrOut}.
    
    \paragraph{\ref{item:1INthm:1InOrOut}.}
    Theorem~\ref{thm:1InOrOut}\ref{item:1INthm:1InOrOut} follows from Lemma~\ref{lem:canonical} using $ 12 $ flips and Lemma~\ref{lem:main} using $ 2 ( n - 6 ) $ flips for a total of $ 2 n $ flips.
    
    \paragraph{\ref{item:2INthm:1InOrOut}.}
    First note that, any flip performed in the proof of Lemma~\ref{lem:main} preserves the degree of $ \xi $.
    In the following, we show that all the canonical paths such that the degree of $ \xi $ is $ 1 $ are connected using only flips that preserve the degree of $ \xi $.

    Let $ p _ 1 , \dots , p _ { n - 1 } $ be the points of $ C $ in counterclockwise order with indices modulo $ n - 1 $, and $ \{ u , \dots , v \} $ be the integer interval of indices such that the segment $ \xi p _ i $ does not intersect the interior of the convex hull of $ C $ (Figure~\ref{fig:extraCanonical}(a)).
    Let $ P _ { i , j } $ be the canonical path such that the degree of $ \xi $ is $ 1 $, such that $ \xi $ is adjacent to $ p _ i $, and such that $ p _ j $ is an extremity of $ P _ { i , j } $.
    
    The set of all the canonical paths such that the degree of $ \xi $ is $ 1 $ consists of all the $ \{ P _ { i , i + 1 } , P _ { i , i - 1 } $ such that $ i \in \{ u , \dots , v \} $. 
    Thus, the following two assertions are enough to conclude the proof.
    \begin{itemize}
        \item For all $ i \in \{ u , \dots , v \} $,  $ P _ { i , i + 1 } $ and $ P _ { i , i - 1 } $ are connected by the flip replacing the segment $ p _ { i - 1 } p _ i $ with $ p _ i p _ { i + 1 } $ (Figure~\ref{fig:extraCanonical}(b) and (c)).
        \item For all $ i \in \{ u , \dots , v - 1 \} $,  $ P _ { i , i + 1 } $ and $ P _ { i + 1 , i } $ are connected by the flip replacing the segment $ \xi p _ i $ with $ \xi p _ { i + 1 } $ (Figure~\ref{fig:extraCanonical}(b) and (d)).
    \end{itemize}
    The diameter of the subgraph of the flip graph is $2 ( n - 6 ) + 2 ( n - 2 ) + 1 = 4 n - 15 $ flips.
    
    \begin{figure}[htb]
        \centering
        \begin{tabularx}{\textwidth}{CCCC}
            \includegraphics[scale=\graphicsScale,page=13]{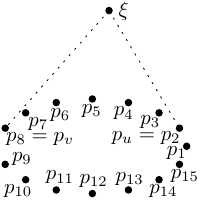}&%
            \includegraphics[scale=\graphicsScale,page=14]{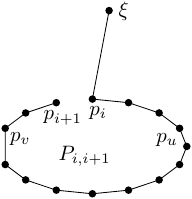}&%
            \includegraphics[scale=\graphicsScale,page=15]{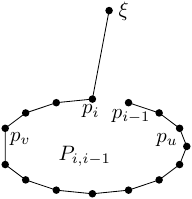}&%
            \includegraphics[scale=\graphicsScale,page=16]{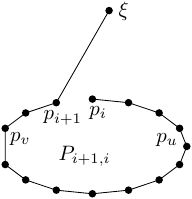}\\
            (a) & (b) & (c) & (d)%
        \end{tabularx}
        \caption{Illustration of the proof of Theorem~\ref{thm:1InOrOut}\ref{item:2INthm:1InOrOut}. (a) The points of $ C $ in counterclockwise order with indices modulo $ n - 1 $, with $ u = 2 $ and $ v = 8 $. (b) The canonical path $ P _ { i , i + 1 } $. (c) The canonical path $ P _ { i , i - 1 } $. (d) The canonical path $ P _ { i + 1 , i } $.}
        \label{fig:extraCanonical}
    \end{figure}
    
    \paragraph{\ref{item:3INthm:1InOrOut}.}
    If the degree of $ \xi $ is $ 2 $, then any flip used in the proof of Theorem~\ref{thm:1InOrOut}\ref{item:1INthm:1InOrOut} preserves the degree of $ \xi $.
\end{proof}

\section{Connected Components of the Flip Graph}

In this section, we successively prove that no connected component of the flip graph $ \FlipGraph { S } $ has only $ 1 $ or $ 2 $ vertices.
Informally, this properties may be seen as failed attempts to disprove Conjecture~\ref{cjt:main}.

\begin{lemma}
    \label{lem:noConnectedComponentSize1}
    If $ S $ contains at least $ 3 $ points, then $ \FlipGraph { S } $ has no isolated vertex.
\end{lemma}
\begin{proof}
    Let $ S $ be any set of at least $ 3 $ points in general position.
    Let $ P = p _ 1 , \dots , p _ n $ be an arbitrary plane path.
    Let $ p _ i $ be a point of the plane path such that the segment $ p _ 1 p _ i $ do not cross the plane path $ P $.
    By Lemma~\ref{lem:see2endpoints}, there exists at least two endpoints $ p $ of segments in $ P $ such that the open segment $ p _ 1 p $ does not intersect any closed segment of $ P \setminus \{ p _ 1 p _ 2 \} $. The point $ p = p _ 2 $ being one of them, there exists at least one point $ p _ i \in S \setminus \{ p _ 1 , p _ 2 \} $ such that the open segment $ p _ 1 p _ i $ does not intersect any closed segment of $ P $ (thanks to the $ S $ being in general position).
    Thus, there exists a flip that replaces one of the two segments incident to $ p _ i $ by the segment $ p _ 1 p _ i $ (see Figure~\ref{fig:lem iso vertex}).
    More specifically, the removed segment is $ p _ { i - 1 } p _ i $ and the resulting path is $ P ' = p _ { i - 1 } , p _ { i - 2 } , \dots , p _ 1 , p _ i , p _ { i + 1 } , \dots , p _ n $.
    Therefore, $ \FlipGraph { S } $ has no isolated vertex.
\end{proof}

\begin{figure}[htb]
    \centering
    \begin{tabularx}{\textwidth}{C}
        \includegraphics[scale=\graphicsScale]{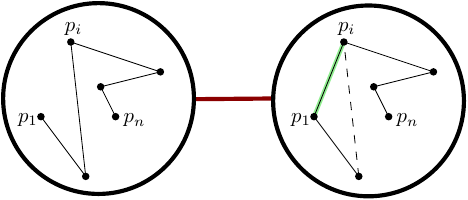}
    \end{tabularx}
 \caption{Illustration of the proof of Lemma~\ref{lem:noConnectedComponentSize1}.}
 \label{fig:lem iso vertex}
\end{figure}

\begin{lemma}
    \label{lem:noConnectedComponentSize2}
    If $ S $ contains at least $ 3 $ points, then $ \FlipGraph { S } $ has no connected component of $ 2 $ vertices.
\end{lemma}
\begin{proof}
    Let $ S $ be any set of at least $ 3 $ points in general position.
    Let $ P $ be an arbitrary plane path of $ S $.
    We perform the same flip as in the proof of Lemma~\ref{lem:noConnectedComponentSize1} for each extremity of $ P $.
    We obtain the paths $ P ' = p _ { i - 1 } , p _ { i - 2 } , \dots , p _ 1 , p _ i , p _ { i + 1 } , \dots , p _ n $ and $ P '' = p _ 1 , p _ 2 , \dots , p _ j , p _ n , p _ { n - 1 } , \dots , p _ { j + 1 } $.
    The paths $ P , P ' , P '' $ are pairwise distinct because $ p _ { i - 1 } \neq p _ 1 $ and $ p _ { j + 1 } \neq p _ n $.
    Therefore, the connected component of $ P $ has at least $ 3 $ vertices.
\end{proof}

\bibliography{ref}

\begin{thebibliography}{10}

\bibitem{aichholzer2022reconfiguration}
Oswin Aichholzer, Brad Ballinger, Therese Biedl, Mirela Damian, Erik~D Demaine,
  Matias Korman, Anna Lubiw, Jayson Lynch, Josef Tkadlec, and Yushi Uno.
\newblock Reconfiguration of non-crossing spanning trees.
\newblock {\em arXiv preprint arXiv:2206.03879}, 2022.

\bibitem{Flipping}
Oswin Aichholzer, Kristin Knorr, Wolfgang Mulzer, Johannes Obenaus, Rosna Paul,
  and Birgit Vogtenhuber.
\newblock Flipping plane spanning paths.
\newblock In {\em WALCOM: Algorithms and Computation: 17th International
  Conference and Workshops}, pages 49--60. Springer, 2023.
\newblock URL: \url{https://arxiv.org/pdf/2202.10831.pdf}.

\bibitem{On}
Selim~G. Akl, Md.~Kamrul Islam, and Henk Meijer.
\newblock On planar path transformation.
\newblock In {\em Information Processing Letters}, volume 104(2), pages 59--64.
  Elsevier, 2007.
\newblock URL: \url{https://doi.org/10.1016/j.ipl.2007.05.009}.

\bibitem{longestuntangleconf}
Guilherme~D. da~Fonseca, Yan Gerard, and Bastien Rivier.
\newblock On the longest flip sequence to untangle segments in the plane.
\newblock In {\em International Conference and Workshops on Algorithms and
  Computation (WALCOM 2023)}, volume 13973 of {\em Lecture Notes in Computer
  Science}, pages 102--112, 2023.
\newblock URL: \url{https://arxiv.org/abs/2210.12036}, \href
  {https://doi.org/10.1007/978-3-031-27051-2_10}
  {\path{doi:10.1007/978-3-031-27051-2_10}}.

\bibitem{hernando2002graphs}
Carmen Hernando, Ferran Hurtado, and Marc Noy.
\newblock Graphs of non-crossing perfect matchings.
\newblock {\em Graphs and Combinatorics}, 18:517--532, 2002.

\bibitem{houle2005graphs}
Michael~E Houle, Ferran Hurtado, Marc Noy, and Eduardo Rivera-Campo.
\newblock Graphs of triangulations and perfect matchings.
\newblock {\em Graphs and Combinatorics}, 21:325--331, 2005.

\bibitem{KKR24}
Linda Kleist, Peter Kramer, and Christian Rieck.
\newblock On the connectivity of the flip graph of plane spanning paths.
\newblock In {\em Graph-Theoretic Concepts in Computer Science (WG 2024)},
  2024.

\bibitem{lawson1972transforming}
Charles~L Lawson.
\newblock Transforming triangulations.
\newblock {\em Discrete mathematics}, 3(4):365--372, 1972.

\bibitem{lubiw2015flip}
Anna Lubiw and Vinayak Pathak.
\newblock Flip distance between two triangulations of a point set is
  np-complete.
\newblock {\em Computational Geometry}, 49:17--23, 2015.

\bibitem{nichols2020transition}
Torrie~L Nichols, Alexander Pilz, Csaba~D T{\'o}th, and Ahad~N Zehmakan.
\newblock Transition operations over plane trees.
\newblock {\em Discrete Mathematics}, 343(8):111929, 2020.

\bibitem{nishimura2018introduction}
Naomi Nishimura.
\newblock Introduction to reconfiguration.
\newblock {\em Algorithms}, 11(4):52, 2018.

\bibitem{OdW07}
Yoshiaki Oda and Mamoru Watanabe.
\newblock The number of flips required to obtain non-crossing convex cycles.
\newblock In {\em Kyoto International Conference on Computational Geometry and
  Graph Theory}, pages 155--165, 2007.

\bibitem{pilz2014flip}
Alexander Pilz.
\newblock Flip distance between triangulations of a planar point set is
  apx-hard.
\newblock {\em Computational Geometry}, 47(5):589--604, 2014.

\bibitem{pournin2014diameter}
Lionel Pournin.
\newblock The diameter of associahedra.
\newblock {\em Advances in Mathematics}, 259:13--42, 2014.

\bibitem{van2013complexity}
Jan van~den Heuvel.
\newblock The complexity of change.
\newblock {\em Surveys in combinatorics}, 409(2013):127--160, 2013.

\bibitem{VLSC81}
Jan van Leeuwen and Anneke~A. Schoone.
\newblock Untangling a traveling salesman tour in the plane.
\newblock In {\em 7th Workshop on Graph-Theoretic Concepts in Computer
  Science}, 1981.

\bibitem{newtrees}
H\aa vard Bakke~Bjerkevik, Linda Kleist, Torsten Ueckerdt, and Birgit
  Vogtenhuber.
\newblock Flipping non-crossing spanning trees.
\newblock personal communication.

\end{thebibliography}

\appendix

\section{Technical Lemma}

In this section, we recall the statement and the proof an important technical lemma from~\cite{On}.
This lemma is used in the proof of Theorem~\ref{thm:convex}~\cite{On}, in~\cite{Flipping}, and in the present article to prove Lemma~\ref{lem:main}.
We add extra technical details to the statement which are used in the proof Lemma~\ref{lem:main}.

\begin{lemma}[Lemma~3.1 in~\cite{On}]
    \label{lem:convex}
    For any set $ S $ of $ n \geq 3 $ points in convex position and any $ P = p _1, \dots, p _ n  $ in $ \Vertices {  \FlipGraph { S } } $, there exists a flip path of length at most $ n - 3 $ in $ \FlipGraph { S } $ connecting $ P $ to a plane path $ P ' $ contained in the convex hull boundary of $ S $ and such that $ P $ and $ P ' $ have a common extremity $ p _ n $ and a common segment $ p _ { n - 1 } p _ n $ incident to this endpoint.

    Moreover, every flip of this flip path does not remove any edge of the convex hull of $ S $.
\end{lemma}
\begin{proof}
    We iteratively perform the following flip, preserving the segment $ p _ { n - 1 } p _ n $ and the fact that $ p _ n $ is an extremity of $ P $.
    Let $ i $ be the smallest index such that the segment $ p _ i p _{ i + 1 } $ is not on the convex hull boundary of the point set.
    Replace the segment $ p _ i p _ { i + 1 } $ with the segment $ p _ 1 p _ { i + 1 } $.
\end{proof}

\section{Geometric Lemmas}

This section gathers the geometric lemmas we use.
Note that we do not assume general position.
Before stating Lemma~\ref{lem:see1endpoint}, we give a definition.
Let $ P $ be a set of subsets of the plane, and $ \mya , b $ be two points. 
We say that $ \mya $ and $ b $ \emph{see} each other if the open segment $ \mya b $ does not intersect $ \bigcup P $.

\begin{lemma}
    \label{lem:see1endpoint}
    Let $ P $ be a non-empty finite set of closed segments.
    If any two segments of $ P $ are either disjoint or intersect at a common endpoint, then any point $ q $ not in $ \bigcup P $ sees at least $ 1 $ endpoint of some segment in $ P $.
\end{lemma}
\begin{proof}
    First, we define some sets; these definitions are summarized in Figure~\ref{fig:see1endpoint}(a).
    Let $ E $ be the set of endpoints of the segments in $ P $.
    Let $ P ' $ be the set of sub-segments of $ P $ seen by $ q $ and
    $ E ' $ be the set of endpoints of the segments in $ P ' $.
    Note that a segment in $ P ' $ is not necessarily a closed segment, and thus that a point in $ E ' $ is not necessarily in $ \bigcup P ' $.
    
    Now, we have the inclusion $ ( \bigcup P ' ) \cap E ' \subseteq E $ because the existence of a point in $ ( ( \bigcup P ' ) \cap E ' ) \setminus E $ contradicts that any intersection of two segments of $ P $ is either empty or a common endpoint.
    Thus, to prove Lemma~\ref{lem:see1endpoint}, it is enough to show that $ ( \bigcup P ' ) \cap E ' $ is not empty.
    Next, we define a point $ p $ in $ E ' $, then prove that this point $ p $ is also in $ \bigcup P ' $.
    
    Since $ P $ is not empty, $ \bigcup P ' $, $ P ' $, and $ E ' $, are also not empty.
    Thus, there exists a point $ p $ in the non-empty closed set $ E ' $ which is the closest to $ q $.
    
    By definition of $ P ' $, the following implication holds.
    If a point $ q ' \in E ' $ is not in $ \bigcup P ' $, then the segment $ q q ' $ contains a point in $ E ' $ (in fact in $ ( \bigcup P ' ) \cap E ' $).
    The contraposition of the previous implication shows that $ p $ is in $ \bigcup P ' $.

    As a conclusion, $ p \in ( \bigcup P ' ) \cap E ' \subseteq E $. In particular, $ p $ is an endpoint of a segment in $ P $ which is seen by $ q $.
\end{proof}

\begin{figure}[htb]
    \centering
    \begin{tabularx}{\textwidth}{CC}
        \includegraphics[scale=\graphicsScale,page=1]{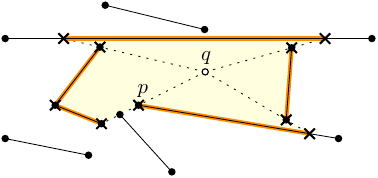}&%
        \includegraphics[scale=\graphicsScale,page=2]{graphics/see1endpoint.pdf}\\
        (a) & (b)%
    \end{tabularx}
    \caption{%
        (a) Illustration of the proof of Lemma~\ref{lem:see1endpoint}. The point $ q $ is drawn with a hollow disc. The segments in $ P $ are drawn plain and thin. The endpoints in $ E $ are drawn with filled discs. The segments in $ P ' $ are highlighted. The endpoints in $ E ' $ are drawn with crosses.
        (b) Illustration of the proof of Lemma~\ref{lem:see2endpoints}. The set $ H $ is shaded. The segments in $ P ' $ are bold.}
    \label{fig:see1endpoint}
\end{figure}

\begin{lemma}
    \label{lem:see2endpoints}
    Let $ P $ be a non-empty finite set of closed segments.
    If any two segments of $ P $ are either disjoint or intersect at a common endpoint, then any point $ q $ not in $ \bigcup P $ sees at least $ 2 $ endpoints of some segments in $ P $.
\end{lemma}
\begin{proof}
    By Lemma~\ref{lem:see1endpoint}, there exists a segment $ p _ 1 p _ 2 $ in $ P $ such that at least one of its endpoints $ p _ 1 , p _ 2 $, say $ p _ 1 $, is seen by $ q $.
    Let $ H $ be set of points $ p $ such that the open segment $ p q $ intersects the closed segment $ p _ 1 p _ 2 $ (informally, $ H $ is the set of points ``hidden'' from $ q $ by the segment $ p _ 1 p _ 2 $).
    Let $ P ' = \{ s \in P \setminus \{ p _ 1 p _ 2 \} : s \cap H = \emptyset \} $.
    Figure~\ref{fig:see1endpoint}(b) illustrates the sets $ H $ and $ P ' $
    There are two cases.
    \begin{Cases}
        \item If $ P ' $ is empty, then $ p _ 2 $ is also seen by $ q $.
        \item Otherwise, Lemma~\ref{lem:see1endpoint} applied to $ P ' $ and $ q $ ensures that there exists an endpoint of a segment in $ P ' $ which is seen by $ q $.\qedhere
    \end{Cases}
\end{proof}

\begin{lemma}
    \label{lem:see3endpoints}
    Let $ P $ be a non-empty finite set of closed segments.
    If any two segments of $ P $ are either disjoint or intersect at a common endpoint, then, for any point $ q $ not in $ \bigcup P $, exactly one of the following holds.
    \begin{thmEnumerate}
        \item \label{item:seeExactly2endpoints} The points in $ \bigcup P $ seen by $ q $ consist of exactly one segment of $ P $ (including its $ 2 $ endpoints).
        \item The points in $ \bigcup P $ seen by $ q $ include at least $ 3 $ endpoints of some segments in $ P $.
    \end{thmEnumerate}
\end{lemma}
\begin{proof}
    Assume that~\ref{item:seeExactly2endpoints} does not hold.
    By Lemma~\ref{lem:see1endpoint}, there exists a segment $ p _ 1 p _ 2 $ in $ P $ such that at least one of its endpoints $ p _ 1 , p _ 2 $, say $ p _ 1 $, is seen by $ q $.
    Let $ H $ be set of points $ p $ such that the closed segment $ p _ 1 p _ 2 $ intersects the open segment $ p q $.
    Let $ P ' = \{ s \in P \setminus \{ p _ 1 p _ 2 \} : s \cap H = \emptyset \} $.
    Similarly to the proof of Lemma~\ref{lem:see2endpoints}, there are two cases.
    \begin{Cases}
        \item If $ P ' $ is empty, then $ p _ 2 $ is also seen by $ q $, and we exhibit a third point seen by $ q $.
        Let $ P '' = \{ s \in P \setminus \{ p _ 1 p _ 2 \} : s \not \subseteq H \} $.
        The assumption that the points in $ \bigcup P $ seen by $ q $ do not consist of exactly one segment of $ P $ implies that $ P '' $ is not empty.
        
        Let $ r _ 1 $ (respectively $ r _ 2 $) be the ray included in the line $ q p _ 1 $ (respectively $ q p _ 2 $) which starts from $ p _ 1 $ (respectively $ p _ 2 $) and does not contain $ q $.
        Let $ P _ 1 $ (respectively $ P _ 2 $) be the set of segments in $ P '' $ which intersect the ray $ r _ 1 $ (respectively $ r _ 2 $).
        As $ P _ 1 \cup P _ 2 = P '' $, at least one of $ P _ 1 , P _ 2 $, say $ P _ 1 $, is not empty.
        Let $ p _ 3 p _ 4 $ be the segment in $ P _ 1 $ whose intersection $ i $ with $ r _ 1 $ is the closest to $ q $.
        At least one of $ p _ 3 , p _ 4 $, say $ p _ 3 $, is not in $ H $.
        
        Since the open segments $ q i $ and $ i p _ 3 $ do not intersect any segment in $ P \setminus \{ p _ 1 p _ 2 \} $, any segment $ s $ in $ P $ intersecting $ q p _ 3 $ has one of its endpoint in the triangle $ q i p _ 3 $ and the other endpoint in the half-plane defined by the line $ q p _ 3 $ and not containing $ p _ 1 $.
        Such a segment $ s $ would therefore be in $ P ' $ which is empty, proving that $ p _ 3 $ is seen by $ q $.
        
        \item Otherwise, Lemma~\ref{lem:see2endpoints} applied to $ P ' $ and $ q $ ensures that there exist $ 2 $ endpoints of some segment in $ P ' $ which is seen by $ q $.\qedhere
    \end{Cases}
\end{proof}

\end{document}